\documentclass[10pt,a4paper]{article}
\usepackage{authblk}
\usepackage{color}

\usepackage{bbm}
\usepackage{amsmath,amsthm}
\usepackage{alltt, amssymb}
\usepackage[]{algorithm2e}

\newcommand{\bN}{ {\mathbb  N}}
\newcommand{\bZ} { {\mathbb Z}}
\newcommand{\bQ}{ {\mathbb  Q}}
\newcommand{\bC}{ {\mathbb  C}}
\newcommand{\bK}{ {\mathbb  K}}

\newcommand{\bE}{ {\mathbb E}}

\newcommand{\bx}{ {\mathbf x}}
\newcommand{\by}{ {\mathbf y}}
\newcommand{\bz}{ {\mathbf z}}
\newcommand{\bc}{ {\mathbf c}}
\newcommand{\bs}{ {\mathbf s}}
\newcommand{\bt}{ {\mathbf t}}
\newcommand{\bff}{ {\mathbf f}}
\newcommand{\bpa}{ {\boldsymbol \partial}}

\newcommand{\bdelta}{ {\boldsymbol \delta}}

\newcommand{\bu}{ {\mathbf u}}
\newcommand{\bw}{ {\mathbf w}}

\newcommand{\bv}{ {\mathbf v}}

\newcommand{\bm}{ {\mathbf m}}

\newcommand{\ie}{{\it i.e.}}

\newcommand{\bzero}{ {\mathbf 0}}

\newcommand{\TT}{ \operatorname{T}}

\newcommand{\HT}{{\operatorname{HT}}}
\newcommand{\IE}{{\operatorname{IE}}}
\newcommand{\PT}{{\operatorname{PT}}}
\newcommand{\PE}{{\operatorname{PE}}}
\newcommand{\HC}{{\operatorname{HC}}}

\newcommand{\rank}{ \operatorname{rank}}
\newcommand{\pa}{ {\partial}}

\newcommand{\ind}{ \operatorname{ind}}
\newcommand{\In}{ \operatorname{in}}
\newcommand{\sol}{ \operatorname{sol}}

\newcommand{\Drat}{ {\bK(\bx)[\bpa]}}
\newcommand{\Dpol}{ {\bK[\bx][\bpa]}}

\def\rank{\operatorname{rank}}

\def\ind{\operatorname{ind}}

\newtheorem{thm}{Theorem}[section]
\newtheorem{cor}[thm]{Corollary}
\newtheorem{lemma}[thm]{Lemma}
\newtheorem{prop}[thm]{Proposition}
\newtheorem{defn}[thm]{Definition}
\newtheorem{ex}[thm]{Example}
\newtheorem{algo}[thm]{Algorithm}

\begin{document}

\title{Apparent Singularities of~D-finite~Systems}

\author[1]{Shaoshi Chen\thanks{Supported by the NSFC grant 11501552 and
by the President Fund of the Academy of
Mathematics and Systems Science, CAS (2014-cjrwlzx-chshsh).  Email: schen@amss.ac.cn}}
\author[2]{Manuel Kauers\thanks{Supported by the Austrian Science Fund (FWF): F50-04 and Y464-N18. Email: manuel@kauers.de}}
\author[3]{Ziming Li\thanks{Supported by the NSFC grants (91118001, 60821002/F02)
and a 973 project (2011CB302401). Email: zmli@mmrc.iss.ac.cn}}
\author[4]{Yi Zhang\thanks{Supported by the Austrian Science Fund (FWF): Y464-N18 and P29467-N32. Email: zhangy@amss.ac.cn}}
\affil[1,3]{KLMM, AMSS, Chinese Academy of Sciences, Beijing, China}
\affil[2, 4]{Institute for Algebra, Johannes Kepler University Linz, Austria}
\affil[4]{RICAM, Austrian Acedemy of Sciences, Austria}

\renewcommand\Authands{ and }


\date{}
\maketitle

\begin{abstract}
We generalize the  notions of singularities and ordinary points from linear ordinary differential equations to D-finite systems.
Ordinary points of a D-finite system are characterized in terms of its formal power series solutions.
We also show that apparent singularities can be removed like in the univariate
case by adding suitable additional solutions to the system at hand.
Several algorithms are presented for removing and detecting apparent singularities.
In addition,  an algorithm is given for computing formal power series solutions of a D-finite system
at apparent singularities.  
\end{abstract}

\section{Introduction}\label{SECT:intro}

Ordinary linear differential equations allow easy access to the singularities of
their solutions: every point $\alpha$ which is a singularity of some solution
$f$ of the differential equation must be a zero of the coefficient of the
highest order derivative appearing in the equation, or a singularity of one of
the other coefficients. For example, $x^{-1}$~is a solution of the equation $x
f'(x) + f(x)=0$, and the singularity at $0$ is reflected by the root of the
polynomial~$x$ in front of the term $f'(x)$ in the equation. Unfortunately, the
converse is not true: there may be roots of the leading coefficient which do not
indicate solutions that are singular there. For example, all the solutions of
the equation $x f'(x) - 5 f(x) = 0$ are constant multiples of~$x^5$, and none of
these functions is singular at~$0$.

For a differential equation $p_0(x)f(x) + \cdots + p_r(x)f^{(r)}(x)=0$ with
polynomial coefficients $p_0,\dots,p_r$ and $p_r\neq0$, the roots of $p_r$ are called the
singularities of the equation. Those roots $\alpha$ of $p_r$ such that the
equation has no solution that is singular at~$\alpha$ are called
\emph{apparent.}  In other words, a root $\alpha$ of $p_r$ is apparent if the
differential equation admits $r$ linearly independent formal power series solutions in
$x-\alpha$. Deciding whether a singularity is apparent is therefore the same as
checking whether the equation admits a fundamental system of formal power series
solutions at this point. This can be done by inspecting the so-called
\emph{indicial polynomial} of the equation at~$\alpha$:
if there exists a power series
solution of the form $(x-\alpha)^\ell + \cdots$, then~$\ell$ is a root of this polynomial.

When some singularity $\alpha$ of an ODE is apparent, then it is always possible
to construct a second ODE whose solution space contains all the solutions of the
first ODE, and which does not have $\alpha$ as a singularity. This process is
called desingularization. The idea is easily explained. The key observation is
that a point $\alpha$ is a singularity if and only if the indicial polynomial at
$\alpha$ is different from $n(n-1)\cdots(n-r+1)$ or the ODE does not
admit $r$ linearly independent formal power series solutions in
$x-\alpha$. As the indicial polynomial at
an apparent singularity has only nonnegative integer roots, we can bring it into
the required form by adding a finite number of new factors.  Adding a factor
$n-s$ to the indicial polynomial amounts to adding a solution of the form
$(x-\alpha)^s+\cdots$ to the solution space, and this is an easy thing to do
using well-known arithmetic of differential operators.
See~\cite{Abramov2006,Barkatou2015,Chen2016,Ince1926,Max2013} for an
expanded version of this argument and~\cite{Abramov2006,Abramov1999} for analogous algorithms for
recurrence equations.

The purpose of the present paper is to generalize the two facts sketched above to
the multivariate setting. Instead of an ODE, we consider systems of PDEs known as
D-finite systems. For such systems, we define the notion of a singularity in terms
of the polynomials appearing in them (Definition~\ref{DEF:op}). We show in Theorem~\ref{THM:chop} that
a point is a singularity of the system unless it admits a basis of power series
solutions in which the starting terms are as small as possible with respect to
some term order. Then a singularity is apparent if the system admits a full basis of power
series solutions, the starting terms of which are not as small as possible. We then prove in
Theorem~\ref{THM:rmappsin} that apparent singularities can be removed like in the univariate
case by adding suitable additional solutions to the system at hand.
The resulting system will be contained in the Weyl closure~\cite{tsai00} of the original ideal,
but unlike Tsai~\cite{tsai00} we cannot guarantee that it is equal to the Weyl closure.
Based on Theorem~\ref{THM:chop} and Theorem~\ref{THM:rmappsin}, we show
how to remove a given apparent singularity (Algorithms~\ref{ALGO:desingularization1} and~\ref{ALGO:rand}),
 and how to detect whether a given point is an apparent singularity (Algorithm~\ref{ALGO:desingularization2}).
At last, we present an algorithm for computing formal power series solutions of a D-finite system
at apparent singularities.

\section{Preliminaries}\label{SECT:op}

In this section, we recall some notions and conclusions concerning linear partial differential operators, Gr\"obner bases,
formal power series, solution spaces and Wronskians for D-finite systems. We also specify notation to be used in the
rest of this paper.

\subsection{Rings of differential operators} \label{SUBSECT:ringdo}
Throughout the paper, we assume that $\bK$ is a field of characteristic zero and~$n$ is a positive integer.
For instance, $\bK$ can be the field of complex numbers.
Let $\bK[\bx]=\bK[x_1, \ldots, x_n]$ be the ring of usual commutative polynomials over~$\bK$.
The quotient field of $\bK[\bx]$ is denoted by $\bK(\bx)$.
Then we have the \emph{ring of differential operators with rational function
coefficients} $\bK(\bx)[\pa_1, \ldots, \pa_n]$,
in which addition is coefficient-wise and  multiplication is defined by associativity via the
commutation rules
\begin{itemize}
 \item [(i)] $\pa_i \pa_j = \pa_j \pa_i$;
 \item [(ii)] $\pa_i f = f \pa_i + \frac{\pa f}{\pa x_i} \text{ for each } f \in \bK(\bx)$,
\end{itemize}
where $\frac{\pa f}{\pa x_i}$ is the usual derivative of~$f$ with respect to $x_i$, $1 \le i, j \le n$.
This ring is an Ore algebra~\cite{Robertz2014, Salvy1998} and denoted by $\Drat$ for brevity.

Another ring is $\Dpol:=\bK[x_1, \ldots, x_n][\pa_1, \ldots, \pa_n]$, which is a subring of~$\Drat$.
We call it the \emph{ring of differential operators with polynomial coefficients} or
the \emph{Weyl algebra}~\cite[Section 1.1]{Saito1999}.

A left ideal $I$ in $\Drat$ is called \emph{D-finite} if the quotient $\Drat \slash I$
is a finite dimensional vector space over $\bK(\bx)$.
The dimension of $\Drat \slash I$ as a vector
space over $\bK(\bx)$ is called the \emph{rank} of $I$ and denoted by~$\rank(I)$.

For a subset $S$ of $\Drat$, the left ideal generated by $S$ is denoted by~$\Drat S$.
For instance, let~$I = \bQ(x_1, x_2)[\pa_1, \pa_2] \ \{ \pa_1 - 1, \pa_2 - 1 \}$.
Then~$I$ is D-finite because the quotient $\bQ(x_1, x_2)[\pa_1, \pa_2] \slash I$
is a vector space of dimension~$1$ over $\bQ(x_1, x_2)$.
Thus, $\rank(I) = 1$.

\subsection{Gr\"{o}bner bases} \label{SUBSECT:gb}
Gr\"{o}bner bases in $\Drat$
are well known~\cite{Weispfenning1990} and implementations for them are available
for example in the Maple package {\tt Mgfun}~\cite{Chyzak2008} and in the
Mathematica package {\tt HolonomicFunctions.m}~\cite{Christoph2010}.
We briefly summarize some facts about Gr\"{o}bner bases.

We denote by $\TT(\bpa)$ the commutative monoid generated by $\pa_1, \ldots, \pa_n$.
An element of $\TT(\bpa)$ is called a {\em term}. For a vector $\bu=(u_1, \ldots, u_n) \in \bN^n$,
the symbol~$\bpa^\bu$ stands for the term~$\pa_1^{u_1} \cdots \pa_n^{u_n}$. The {\em order} of~$\bpa^\bu$
is defined to be $|\bu| := u_1+ \cdots + u_n$. For a nonzero operator~$P \in \Drat$, the \emph{order} of~$P$
is defined to be the highest order of the terms that appear in~$P$ effectively.

Let $\prec$~\footnote{In examples, we use the graded lexicographic ordering with $\partial_n\succ\dots\succ\partial_1$.}
be a graded monomial ordering~\cite[Definition 1, page 55]{Cox2015} on $\bN^n$.
Since there is a one-to-one correspondence between terms in $\TT(\bpa)$
and elements in $\bN^n$, the ordering~$\prec$ on~$\bN^n$ induces an ordering
on $\TT(\bpa)$ with $\bpa^{\bu} \prec \bpa^{\bv}$ if $\bu \prec \bv$.
Our main results on apparent singularities are based on the fact that there are
 at most finitely many terms lower than a given term. So we fix a graded ordering~$\prec$ on~$\bN^n$ in the rest of the paper.

For a nonzero element $P \in \Drat$, the {\em head term} of~$P$, denoted by~$\HT(P)$, is the highest term
appearing in~$P$. The coefficient of~$\HT(P)$ is called the {\em head coefficient} of $P$ and is denoted by~$\HC(P)$.
For a subset~$S$ of nonzero elements in~$\Drat$, $\HT(S)$ and $\HC(S)$ stand for the sets of head terms and head
coefficients of the elements in~$S$, respectively.

For a Gr\"{o}bner basis $G$ in~$\Drat$, a term is said to be \emph{parametric} if it is not divisible by any term in
$\HT(G)$. The set of exponents of all parametric terms is
referred to as the set of parametric exponents of $G$ and denoted by $\PE(G)$.
If $\Drat G$ is D-finite, then its rank is also called the rank of $G$ and denoted by $\rank(G)$, which is equal to~$|\PE(G)|$.

\subsection{Formal power series} \label{SUBSECT:fps}
Let $\bK[[\bx]]$ be the ring of formal power series with respect to $x_1, \ldots, x_n$.
For $P \in \Dpol$ and $f \in \bK[[\bx]]$, there is a natural action
of $P$ on~$f$, which is denoted by~$P(f)$.
For $P, Q \in \Dpol$, it is straightforward to verify that
\begin{equation} \label{EQ:comm}
P Q(f) = P(Q(f)).
\end{equation}
For~$\bu =(u_1, \ldots, u_n) \in \bN^n$, the product $(u_1 !) \cdots (u_n !)$ is denoted by $\bu!$,  and $x_1^{u_1} \cdots x_n^{u_n}$ by $\bx^\bu$.
A formal power series can always be written in the form
\[ f = \sum_{\bu \in \bN^n} \frac{c_{\bu}}{\bu !} \bx^{\bu}, \]
where $c_{\bu} \in \bK$. Such a form is convenient for differentiation.

Taking the constant term~$c_{{\bf 0}}$ of a formal power series~$f$ gives rise to a ring homomorphism, which is denoted by~$\phi$.
A direct calculation yields
\begin{equation} \label{EQ:zf}
\phi\left( \bpa^{\bu}(f) \right) =c_{\bu}.
\end{equation}
Thus, we can determine whether a formal power series is zero by differentiating
and taking constant terms, as stated in the next lemma.
\begin{lemma} \label{LM:zero}
Let $f \in \bK[[\bx]]$. Then $f = 0$ if and only if, for all $\bu \in \bN^n$,
\[   \phi\left( \bpa^{\bu}(f) \right) = 0. \]
\end{lemma}

The following result appears in~\cite{Gessel1981} for $s=1$. But the proof applies literally
also for arbitrary values of~$s$. Please see \cite[Lemma 4.3.3]{Zhang2017} for a detailed verification.
\begin{lemma} \label{LEM:gessel}
Let $p_1, p_2, \ldots, p_s$ and $q$ be polynomials in $\bK[\bx]$ with
$$\gcd(p_1, p_2, \ldots, p_s, q) = 1.$$
If $p_i \slash q$ has a power series expansion for each $i \in \{1, 2, \ldots, s \}$,
then the constant term of $q$ is nonzero.
\end{lemma}

The fixed ordering $\prec$ on $\bN^n$ also induces an ordering on the monoid  $\TT(\bx)$ generated by $x_1, \ldots, x_n$
in the following manner: $\bx^{\bu} \prec \bx^{\bv}$ if $\bu \prec \bv$.
A nonzero element $ f \in \bK[[\bx]]$ can be written as
\[
 f = \frac{c_{\bu}}{\bu!} \bx^{\bu} + \text{ higher monomials with respect to } \prec,
\]
where~$c_{\bu} \in \bK$ is nonzero.
We call $\bu$
the \emph{initial exponent} of $f$.

\subsection{Solutions and Wronskians} \label{SUBSECT:sw}

Basic facts about solutions of linear partial differential polynomials are presented in~\cite[Chapter IV, Section 5]{Kolchin1973}. We recall them in terms of D-finite ideals.
The first proposition is a special case of Proposition 2 in \cite[page 152]{Kolchin1973}.
\begin{prop} \label{PROP:wn}
For a left ideal~$I \subset \Drat$ with rank~$d$, there exists a differential field~$\bE$ containing~$\bK[[\bx]]$ such that the solution space
of~$I$ has dimension~$d$ over $C_\bE$, where $C_\bE$ stands for the subfield of the constants in~$\bE$.
\end{prop}
Such  differential fields can also be constructed by the Picard-Vessiot approach (see, \cite[Appendix D]{PutSinger} or \cite{BronsteinLiWu}).
In the rest of this paper, we assume that $\bE$ is a differential field described in the above proposition. For a  D-finite ideal~$I$,
the solution space of~$I$ in~$\bE$ is denoted by~$\sol_\bE(I)$. Likewise, for a finite-rank Gr\"obner basis~$G$, the solution space of~$\Drat G$
is simply denoted by $\sol_\bE(G)$.

 The next proposition is an analog of differential Nullstellensatz for D-finite ideals. It is an easy consequence of Corollary~1 in \cite[page 152]{Kolchin1973}.
\begin{prop} \label{PROP:dnull}
  Let~$V \subset \bE$ be a $d$-dimensional linear subspace over~$C_\bE$.
  Then
there exists a unique left ideal~$I \subset \bE[\bpa]$ of rank~$d$ such that
$$V=\sol_\bE(I).$$ Furthermore,
an operator $P$ belongs to~$I$ if and only if~$P$ annihilates every element of~$V$.
\end{prop}

Linear dependence over constants can be determined by Wronskian-like determinants~\cite[Chapter II, Theorem 1]{Kolchin1973}, which implies that a finite number of elements in~$\bK[[\bx]]$ are linearly
independent over~$\bK$ if and only if they are linearly independent over any field of constants that contains~$\bK$.

Wronskian-like determinants are expressed by elements of $\TT(\bpa)$ via wedge notation in \cite{Li2002}.
For $\bv_1, \bv_2, \ldots, \bv_{\ell} \in \bN^n$ and $\ell \in \bZ^{+}$, the exterior product
\[
  \bpa^{\bv_1} \wedge \bpa^{\bv_2} \wedge \cdots \wedge \bpa^{\bv_\ell}
\]
is defined as a multi-linear function from $\bE^{\ell}$ to~$\bE$ that maps~$(z_1, \ldots, z_\ell) \in \bE^{\ell}$ to:
\[
\begin{vmatrix}
\bpa^{\bv_1}(z_1) & \bpa^{\bv_1}(z_2) & \cdots & \bpa^{\bv_1}(z_{\ell}) \\
\bpa^{\bv_2}(z_1) & \bpa^{\bv_2}(z_2) & \cdots & \bpa^{\bv_2}(z_\ell) \\
\vdots        & \vdots        & \ddots & \vdots        \\
\bpa^{\bv_\ell}(z_1) & \bpa^{\bv_\ell}(z_2)& \cdots & \bpa^{\bv_\ell}(z_\ell)
\end{vmatrix}.
\]
It follows from Theorem~1 in~\cite[Chapter II]{Kolchin1973} that~$z_1$, \ldots, $z_\ell$
are linearly independent over~$C_\bE$ if there exist~$\bv_1, \ldots, \bv_\ell \in \bN^n$ such that
\[     (\bpa^{\bv_1} \wedge \cdots \wedge \bpa^{\bv_\ell})(z_1, \ldots, z_\ell) \neq 0. \]

Let~$G$ be a reduced and finite-rank Gr\"obner basis in $\Drat,$ the Wronskian operator of~$G$ is defined to be
\[   w_G := \bigwedge_{\bu \in \PE(G)} \bpa^\bu. \]
The following proposition is Lemma 4 in~\cite{Li2002} in slightly different notation.
\begin{prop} \label{PROP:wr}
Let~$d = \rank(G)$, $z_1, \ldots, z_d \in \sol_\bE(G)$ and~$\PE(G)=\{\bu_1, \ldots, \bu_d\}$.
\begin{itemize}
\item[(i)] The elements $z_1, \ldots, z_d$ are linearly independent over~$C_\bE$ if and only if $w_G(z_1, \ldots, z_d)$ is nonzero.
\item[(ii)] Let~$\bpa^\bv$ be the head term of an element~$g$ of~$G$, and let $z_1, \ldots, z_d$ be linearly independent over~$C_\bE$.
Set~$\bz=(z_1, \ldots, z_d)$ and
\[   w_G \wedge \bpa^\bv(\bz, \cdot) =
\begin{vmatrix}
\bpa^{\bu_1}(z_1) & \bpa^{\bu_1}(z_2) & \cdots & \bpa^{\bu_1}(z_d) & \bpa^{\bu_1} \\
\bpa^{\bu_2}(z_1) & \bpa^{\bu_2}(z_2) & \cdots & \bpa^{\bu_2}(z_d) & \bpa^{\bu_2}\\
\vdots     & \vdots     &        & \vdots     & \vdots \\
\bpa^{\bu_d}(z_1) & \bpa^{\bu_d}(z_2) & \cdots & \bpa^{\bu_d}(z_d) & \bpa^{\bu_d}\\
\bpa^{\bv}(z_1) & \bpa^{\bv}(z_2) & \cdots & \bpa^{\bv}(z_d) & \bpa^{\bv}\\
\end{vmatrix},
\]
in which the elements of~$\TT(\bpa)$ are placed on the right-hand side of a product. Then
\[ w_G(\bz)^{-1} \left( w_G \wedge \bpa^\bv(\bz, \cdot) \right)= \HC(g)^{-1} g. \]
\end{itemize}
\end{prop}
The two propositions listed above will be used to reconstruct a Gr\"obner basis from its solutions.
\section{ Ordinary points and singularities} \label{SECT: ops}
Let $P \in \Dpol$ be in the form
$P = c_{\bu_m} \bpa^{\bu_m} + c_{\bu_{m-1}} \bpa^{\bu_{m-1}} + \cdots + c_{\bu_0} \bpa^{\bu_0},$
where $c_{\bu_0}, \ldots, c_{\bu_m} \in \bK[\bx] \setminus \{0\}$ and $ \bpa^{\bu_0}$ \ldots, $ \bpa^{\bu_ m}$ are distinct.
We say that $P$ is \emph{primitive} if
$\gcd(c_{\bu_0}, c_{\bu_1}, \ldots, c_{\bu_m}) = 1.$

For brevity, a Gr\"{o}bner basis~$G$ in $\Drat$ is said to be \emph{primitive} if it is finite, reduced and its elements are primitive ones in $\Dpol$.
Every nontrivial left ideal in $\Drat$ has a primitive Gr\"obner basis. The goal of this section is to characterize ordinary points
of a primitive Gr\"obner basis of finite rank in terms of formal power series solutions.

\subsection{Ordinary points} \label{SUBSECT:op}

Our definitions of singularities and ordinary points are motivated by the material after~\cite[Lemma 1.4.21]{Saito1999}.

\begin{defn} \label{DEF:op}
Assume that $G$ is a primitive Gr\"obner basis of finite rank.
 A point $ \mathbf{\alpha} \in \overline{\bK}^n$  is called an \emph{ordinary point} of~$G$ if, for every~$p \in \HC(G)$, $p(\mathbf{\alpha}) \neq 0$.
 Otherwise, it is called  a \emph{singularity} of $G$.
\end{defn}

The above definitions are compatible with those in the univariate case~\cite{Abramov2006, Chen2016}.
Note that the origin is an ordinary point of $G$ if and only if each element of $\HC(G)$ has a nonzero constant term.

\begin{ex} \label{EX:op}
Consider the Gr\"{o}bner basis in $\bQ(x_1, x_2)[\pa_1, \pa_2]$
\[
 G = \{\pa_2 - \pa_1, \pa_1^2  + 1 \}.
\]
We find that $\HT(G) = \{\pa_1^2, \pa_2 \}$ and $\HC(G) = \{1 \}$.
So $G$ has no singularity.
\end{ex}

\begin{ex} \label{EX:nop}
Consider the Gr\"{o}bner basis~\cite[Example 3]{Li2002} in $\bQ(x_1, x_2)[\pa_1, \pa_2]$
\[
 G = \{x_1 \pa_1^2  - (x_1 x_2 - 1) \pa_1  - x_2, x_2 \pa_2 - x_1 \pa_1 \}.
\]
In this case, $\HT(G) = \{\pa_1^2, \pa_2 \}$ and $\HC(G) = \{x_1, x_2 \}$ and $\PT(G) = \{1, \pa_1 \}$.
The singularities of $G$ are
$$\{ (a, b) \in \overline{\bQ}^2 \mid a = 0 \text{ or } b = 0 \},$$
which are two lines in~$\overline{\bQ}^2$. In particular, the origin is a singularity.
\end{ex}

\subsection{Characterization of ordinary points} \label{SUBSECT:chop}

From now on, we focus on formal power series solutions of a primitive Gr\"obner basis around the origin,
as a point in $\bK^n$ can always be translated to the origin, and we may assume that $\bK$ is algebraically closed when necessary.
\begin{thm} \label{THM:chop}
Let $G$ be a primitive Gr\"{o}bner basis of finite rank.
Then the origin is an ordinary point of $G$ if and only if
$G$ has $\rank(G)$ many $\bK$-linearly independent formal power series solutions
whose initial exponents are exactly those in $\PE(G)$.
\end{thm}

%
%
%
%
%
%
%

\begin{proof}
Let~$G=\{G_1, \ldots, G_k\}$, $\bpa^{\bv_i} = \HT(G_i)$ and~$\ell_i = \HC(G_i)$, $i=1,\ldots,k$.
Then we can write
\[ G_i = \ell_i \bpa^{\bv_i} + \text{a linear combination of parametric terms over $\bK[\bx]$}  \]

\noindent
\emph{Necessity.} Assume that the origin is an ordinary point of~$G$. Then none of the $\ell_i$'s vanishes at the origin.
We show how to construct formal power series solutions of $G$ by an approach described in~\cite{Wu1989}.

We associate to each tuple $\bu \in \PE(G)$ an arbitrary constant $c_{\bu} \in \bK$.
For a non-parametric term $\bpa^{\bv}$, let
$N_\bv$ be its normal form with respect to~$G$.
Although $N_{\bv}$ belongs to~$\Drat$,
there exists a power product $\ell_{\bv}$ of $\ell_1,\ldots,\ell_k$ such that $\ell_\bv N_{\bv} \in \Dpol$. Write
\[ \ell_\bv(\bx) N_{\bv} = \sum_{\bu \in \PE(G)} a_{\bu, \bv}(\bx) \bpa^\bu \]
with $a_{\bu, \bv} \in \bK[\bx]$. Set
\begin{equation} \label{EQ:ansatz}
c_\bv =  \ell_\bv(\mathbf{0})^{-1} \sum_{\bu \in \PE(G)} a_{\bu, \bv}(\mathbf{0}) c_\bu.
\end{equation}
Note that $\ell_\bv$ can be chosen to be any power product of $\ell_1, \ldots, \ell_k$
such that~$\ell_\bv N_\bv$ belongs to~$\Dpol$.
Let
\[ f = \sum_{\bu \in \bN^n} \frac{c_{\bu}}{\bu !} \bx^{\bu}. \]

 We claim that $f$ is a formal power series solution of~$G$, that is,
 \begin{equation} \label{EQ:claim1}
 G_i(f)=0, \quad i = 1, \ldots, k.
 \end{equation}
 By~\eqref{EQ:comm} and Lemma~\ref{LM:zero}, it suffices to prove

 \begin{equation} \label{EQ:claim2}
 \phi\left( \bpa^{\bu} G_i (f) \right) = 0
 \end{equation}
 for all $\bu \in \bN^n$ and $i \in \{1, \ldots, k\}$.
 We proceed by Noetherian induction on the term order~$\prec$.

 Starting with $\bpa^{\mathbf{0}}$, we can write
 \begin{equation} \label{EQ:initialcase}
  \bpa^{\mathbf{0}}  G_i   = G_i = \ell_i(\bx) \bpa^{\bv_i}- \sum_{\bu \in \PE(G)} a_{\bu, \bv_i}(\bx) \bpa^\bu,
  \end{equation}
 where $a_{\bu, \bv_i} \in \bK[\bx]$. It follows that
 \[ \ell_i(\bx) N_{\bv_i} =  \sum_{\bu \in \PE(G)} a_{\bu, \bv_i}(\bx) \bpa^\bu. \]
 By \eqref{EQ:ansatz},
 \[ \ell_i(\mathbf{0})  c_{\bv_i} - \sum_{\bu \in \PE(G)} a_{\bu, \bv_i}(\mathbf{0}) c_{\bu} = 0, \]
 which can be rewritten as
 \[ \phi(\ell_i(\bx)) \phi(\bpa^{\bv_i}(f)) -  \sum_{\bu \in \PE(G)} \phi( a_{\bu, \bv_i}(\bx)) \phi(\bpa^{\bu}(f)) = 0. \]
 Since $\phi$ is a ring homomorphism, we have
\[  \phi\left(\ell_i(\bx) \bpa^{\bv_i}(f) -  \sum_{\bu \in \PT(G)}  a_{\bu,\bv_i}(\bx) \bpa^{\bu}(f)\right) = 0. \]
We see that $\phi(G_i(f))=0$ by~\eqref{EQ:initialcase}.

Assume that $\bpa^\bv$ is a term higher than $\bpa^{\mathbf{0}}$ and, for all $\bw$ with~$\bw \prec  \bv$ and all $i \in \{1, \ldots, k\}$,
\[  \phi(\bpa^\bw G_i(f)) = 0. \]
Reducing $\bpa^{\bv + \bv_i}$ modulo~$G$, we have
\[  \ell(\bx) \bpa^{\bv + \bv_i} = p_\bv(\bx) (\bpa^\bv G_i) + \left( \sum_{\bw \prec \bv} \sum_{s=1}^k p_{\bw, s}(\bx) (\bpa^\bw G_s) \right) + \ell(\bx) N_{\bv + \bv_i}, \]
where $\ell(\bx)$ and $p_\bv(\bx)$ are two power products of $\ell_1(\bx), \ldots, \ell_k(\bx)$,
and $p_{\bw, s}(\bx)$ belongs to $\bK[\bx]$ for all $\bw \prec \bv$ and~$s \in \{1, \ldots, k\}$.
Moreover, $\ell(\bx) N_{\bv + \bv_i}$ belongs to~$\Dpol$.
Applying the above equality to~$f$, we get
\[  \ell(\bx) \bpa^{\bv + \bv_i}(f) = p_\bv(\bx) (\bpa^\bv G_i)(f) + \left( \sum_{\bw \prec \bv} \sum_{s=1}^k p_{\bw, s}(\bx) (\bpa^\bw G_s)(f) \right) + \ell(\bx) N_{\bv + \bv_i}(f). \]
Applying $\phi$ to the above equality yields
\[
\begin{array}{rcl}
  \phi \left( \ell(\bx) \bpa^{\bv + \bv_i} (f) \right) & =  & p_\bv(\mathbf{0}) \phi(\bpa^\bv G_i(f)) + \sum_{\bw \prec \bv} \sum_{s=1}^k p_{\bw, s}(\mathbf{0}) \phi(\bpa^\bw G_s(f)) \\ \\
  &  & +  \phi\left( \left(\ell(\bx) N_{\bv + \bv_i}\right)(f)\right).
\end{array}
\]
By the induction hypothesis, $\phi(\bpa^\bw G_s(f))=0$ for all $\bw$ with $\bw \prec \bv$ and for all $s \in \{1, \ldots, k\}$. Thus,
\[ \phi\left(\ell(\bx) \bpa^{\bv + \bv_i} (f)\right) = p_\bv(\mathbf{0}) \phi(\bpa^\bv G_i(f))  + \phi\left( \left(\ell(\bx) N_{\bv + \bv_i}\right)(f)\right). \]
Writing~$\ell(\bx) N_{\bv + \bv_i} = \sum_{\bu \in \PE(G)} a_{\bu, \bv + \bv_i}(\bx) \bpa^\bu$ with  $a_{\bu, \bv + \bv_i}(\bx) \in \bK[\bx]$, we see that the above equality implies
\[ \ell(\mathbf{0}) c_{\bv + \bv_i} = p_\bv(\mathbf{0}) \phi(\bpa^\bv G_i(f))  +  \sum_{\bu \in \PE(G)} a_{\bu, \bv + \bv_i}(\mathbf{0}) c_{\bu}. \]
It follows from \eqref{EQ:ansatz} that
\[ p_\bv(\mathbf{0}) \phi(\bpa^\bv G_i(f)) = 0. \]
Since $p_\bv(\mathbf{0})$ is nonzero, $\phi(\bpa^\bv G_i(f))$ is equal to zero. This proves~\eqref{EQ:claim2}.
Therefore, our claim~\eqref{EQ:claim1} holds.
Since there are $\rank(G)$ many parametric terms,
the D-finite system $G$ has $\rank(G)$ many~$\bK$-linearly
independent formal power series solutions with initial exponents in $\PE(G)$.




\medskip \noindent
\emph{Sufficiency.}  Let~$d=\rank(G)$.
Assume that $f_1, \ldots, f_d$
are $\bK$-linearly independently formal power series solutions of $G$ whose
the initial exponents are $\bu_1$, \ldots, $\bu_d$, respectively.
Without loss of generality, we assume further that, for all~$j, m \in \{1, \ldots, d\}$,
\[ \phi(\bpa^{\bu_m}(f_j)) = \delta_{mj},\]
where~$\delta_{mj}$ stands for Kronecker's symbol.

Let $\bff = (f_1, \ldots, f_d)$.
By the above assumption, the constant term of $w_G(\bff)$ is nonzero.
So the formal power series $w_G(\bff)$ is invertible in $\bK[[\bx]]$.

Let $F_i = (w_L \wedge \bpa^{\bv_i})(\bff, \cdot)$. By Proposition~\ref{PROP:wr},
\begin{equation} \label{EQ:wronrep}
\frac{1}{\ell_i} G_i = w_G(\bff)^{-1} F_i \in \bK[[\bx]][\bpa].
\end{equation}
Since $G_i$ is primitive, we can write $G_i$ as
\[
 \ell_i \bpa^{\bv_i} + \sum_{j = 1}^d \ell_{ij} \bpa^{\bu_j},
\]
where $\ell_{ij} \in \bK[\bx]$ and $\gcd(\ell_i, \ell_{i1}, \ldots, \ell_{id}) = 1$.
By~\eqref{EQ:wronrep}, we have
$$\frac{\ell_{ij}}{\ell_i} \in \bK[[\bx]] \quad \text{ for each}  \quad j = 1, \ldots, d.$$
It follows from Lemma~\ref{LEM:gessel} that
the constant term of $\ell_i$ is nonzero. Hence, the origin is an ordinary point of $G$.
\end{proof}

The proof for the necessity of the above theorem also holds for an arbitrary
left (not necessarily D-finite) ideal $\Drat G$,
provided that the origin is an ordinary point of $G$. In addition,
the above theorem also holds when the fixed ordering $\prec$ is not graded. But the results
in the next section hinge on the assumption that $\prec$ is graded.

\section{Apparent singularities}

The goal of this section is to define apparent singularities of a primitive Gr\"obner basis of finite rank, and to characterize them.

\begin{defn} \label{DEF:appsin}
Let $G$ be a primitive Gr\"obner basis of rank $d$.
\begin{itemize}
 \item [(i)] Assume that the origin is a singularity of~$G$.
 We call the origin an \emph{apparent singularity} of $G$ if
 $G$ has $d$ linearly independent formal power series solutions over $\bK$.
 \item[(ii)] Assume that $M$ is a primitive Gr\"obner basis of finite rank.
 We call~$M$ a \emph{left multiple} of $G$ if
 $$\Drat M \subset \Drat G.$$
\end{itemize}
\end{defn}

The above definition is compatible with the univariate case~\cite[Definition 5]{Abramov2006}.

\begin{ex} \label{EX:appsin1}
The solution space $\sol_{\bE}(G)$ of the Gr\"obner basis
\[
 G = \{x_2 \pa_2 + \pa_1 - x_2 - 1, \pa_1^2 - \pa_1 \}
\]
in $\bK(x_1, x_2)[\pa_1, \pa_2]$ is generated by $\{\exp(x_1 + x_2), x_2 \exp(x_2)\}$.
In this case,
$$\HT(G) = \{\pa_2, \pa_1^2 \}, \quad \HC(G) = \{x_2, 1 \} \quad \text{ and } \quad \PE(G) = \{(0,0), (1,0)\}.$$
Therefore,
the origin is a singularity of $G$.  As $G$ has two $\bK$-linearly independent
formal power series solutions, the origin is an apparent singularity of~$G$.

Let $M$ be another Gr\"{o}bner basis such that
\[
 \Drat \cdot M = \Drat G \cap \Drat \cdot \{x_1 \pa_1 - 1, \pa_2 \}.
\]
We find that $M$ is a left multiple of~$G$ with rank $3$.
\end{ex}

\begin{ex} \label{EX:appsin2}
The solution space $\sol_{\bE}(G)$ of the Gr\"obner basis
\[
 G = \{x_2^2 \pa_2 - x_1^2 \pa_1 + x_1 - x_2, \pa_1^2 \}
 \]
in $\bK(x_1, x_2)[\pa_1, \pa_2]$ is generated by $\{x_1 + x_2, x_1 x_2\}$.
In this case,
$$\HT(G) = \{\pa_2, \pa_1^2 \}, \HC(G) = \{x_2^2, 1 \} \text{ and } \PE(G) = \{(0,0),(1,0)\}.$$
Therefore,
the origin is an apparent singularity of $G$.

Set
$$S = \{ (0, 0), (0, 1), (2, 0), (0, 2) \}.$$
Let $M$ be another Gr\"{o}bner basis with
\[
 \Drat \cdot M = \Drat G \cap \left(\bigcap_{(s, t) \in S} \Drat \cdot \{ x_1 \pa_1 - s, x_2 \pa_2 - t \} \right)
\]
We find that $rank(M) = 6$.
By Definition~\ref{DEF:appsin} (ii), $M$ is a left multiple of~$G$ with~$\rank(M)=6$.
\end{ex}

For a subset~$S$ of~$\Drat$, we denote by $\IE_\bzero(S)$ the set of initial exponents of nonzero elements in $\sol_\bE(S) \cap \bK[[\bx]]$ and call it the set of initial exponents of
$S$ at the origin. Then $|\IE_\bzero(S)|$ is the dimension of~$\sol_\bE(S) \cap \bK[[\bx]]$, because any set of formal power series with distinct initial exponents are linearly independent over~$\bK$.
For a primitive Gr\"obner basis~$G$,
the origin is an ordinary point of~$G$ if and only if $\IE_\bzero(G)=\PE(G)$ by Theorem~\ref{THM:chop}, it is an apparent singularity if and only if~$\IE_\bzero(G) \neq \PE(G)$ but $|\IE_\bzero(G)| = |\PE(G)|$ by Definition~\ref{DEF:appsin}.

Before characterizing apparent singularities, we prove two lemmas.
The results in the first lemma are likely known, but we were not able to find proper references containing them.
\begin{lemma} \label{LEM:exact}
Let $I$ and $J$ be $D$-finite ideals in $\Drat$. Then
\begin{itemize}
\item [(i)] $\rank(I \cap J) + \rank(I + J) = \rank(I) + \rank(J).$
\item [(ii)] $\dim \sol_\bE(I \cap J) + \dim \sol_\bE(I + J) = \dim \sol_\bE(I) + \dim \sol_\bE(J).$
\item[(iii)] $\sol_\bE(I \cap J) = \sol_\bE(I) + \sol_\bE(J)$.
\end{itemize}
\end{lemma}
\begin{proof}
Let~$V$ be a vector space over any field, and let $U$ and $ W$ be two subspaces of~$V$.
Set
\[
\begin{array}{llll}
 \psi: & V \slash (U \cap W) & \rightarrow & V \slash U \times V \slash W \\
       &        v + U \cap W & \mapsto     & (v + U, - v + W),
\end{array}
\]
and
\[
\begin{array}{llll}
 \phi: & V \slash U \times V \slash W & \rightarrow & V \slash (U + W) \\
       &               (a + U, b + W) & \mapsto     & a + b + (U + W).
\end{array}
\]
It is straightforward to verify that the following sequence is exact\footnote{We thank Professor Yang Han for bring this exact sequence to our attention, which shortens our original proof.}.
\[
  0 \rightarrow V \slash (U \cap W) \xrightarrow{\psi} V \slash U \times V \slash W \xrightarrow{\phi} V \slash (U + W) \rightarrow 0.\]
It follows that  $V \slash U \times V \slash W$ is linearly isomorphic to $V \slash (U \cap W) \oplus V \slash (U + W).$ In particular,
\[ \dim(V/(U\cap W)) + \dim(V/(U+W)) = \dim(V/U) + \dim(V/W). \]

Setting~$V = \Drat$, $U=I$ and~$W=J$, we prove the first assertion.
The second assertion follows from the first one and Proposition~\ref{PROP:wn}.

For the last assertion, it is evident that $\sol_\bE(I) + \sol_\bE(J) \subset \sol_\bE(I \cap J)$. On the other hand,
\[ \begin{array}{ccl}
\dim (\sol_\bE(I) + \sol_\bE(J)) & = &  \dim (\sol_\bE(I)) + \dim( \sol_\bE(J)) - \dim (\sol_\bE(I) \cap \sol_\bE(J)). \\ \\
 & = & \dim (\sol_\bE(I)) + \dim( \sol_\bE(J)) - \dim (\sol_\bE(I+J)) \\ \\
 &   & \text{(since $\sol_\bE(I) \cap \sol_\bE(J) = \sol_\bE(I+J)$)} \\ \\
 & = & \dim (\sol_\bE(I \cap J)) \quad \text{(by the second assertion).}
\end{array}
\]
Hence, $\sol_\bE(I \cap J) = \sol_\bE(I) + \sol_\bE(J)$.
\end{proof}

As a matter of notation, we define $\bN_m^n =\{ \bu \in \bN \mid |\bu| =m \}$ for $m \in \bN$.
The second lemma illustrates a connection between parametric exponents and initial ones.

\begin{lemma} \label{LEM:wronskianrep}
Let $M$ be a primitive Gr\"{o}bner basis.
Assume that $\sol_{\bE}(M)$ has a basis in $\bK[[\bx]]$ and $\IE_\bzero(M) = \bN^n_m$ for some $m \in \bN$.
Then $\IE_\bzero(M)=\PE(M)$. Consequently, the origin is an ordinary point of $M$.
\end{lemma}
\begin{proof}
Assume that~$f_1, \ldots f_\ell \in \bK[[\bx]]$ form a basis of $\sol_\bE(M)$ and their initial exponents are distinct.
Then $\ell = |\bN_m^n|$.
Let $\bff = (f_1, \ldots, f_{\ell} )$.
And set
\[  w  = \bigwedge_{\bu \in \IE_\bzero(M)}  \bpa^{\bu}. \]
Then~$w(\bff)$ is a nonzero element in $\bK[[\bx]]$.
For every~$\bv \in \bN^n$ with~$|\bv|=m+1$,
let $F_\bv = (w_M \wedge \bpa^\bv)(\bff, \cdot)$, which belongs to~$\bK[[\bx]][\bpa]$.
Then $\HT(F_\bv) = \bpa^\bv$ because $w(\bff)$ is nonzero  and the ordering~$\prec$ is graded.
Since~$F_\bv$ annihilates $f_1, \ldots, f_\ell$, it vanishes
on $\sol_\bE(M)$. It follows from Proposition~\ref{PROP:dnull} that $F_\bv$ belongs to the extended ideal $\bE[\bpa]\cdot M$,
in which~$M$ is still a Gr\"{o}bner basis. Thus, $F_\bv$ can be reduced to zero by~$M$. Accordingly, $\bpa^\bv$ is not
a parametric derivative of~$M$. In other words, $\PE(M)$ is a subset of~$\bN^n_{m}$. Hence,~$\PE(M)=\bN^n_{m}$ because $|\PE(M)|=\ell$ and $\ell= |\bN^n_m|$.
The origin is an ordinary point by Theorem~\ref{THM:chop}.
\end{proof}

\begin{thm} \label{THM:rmappsin}
Let~$G$ be a primitive Gr\"obner basis of rank~$d$. Assume that the origin is a singularity of $G$.
Then the origin is an apparent singularity of~$G$ if and only if it is an ordinary point of some left multiple of~$G$.
\end{thm}

\begin{proof} \emph{Sufficiency.} Assume that the origin is an apparent singularity of~$G$.

Set~$m = \max_{\bu \in \IE_0(G)} |\bu|.$  For every~$\bv =(v_1, \ldots, v_n) \in \bN^n$,
we denote by~$I_\bv$ the left ideal generated by~$x_1 \pa_1 - v_1,$ \ldots, $x_n \pa_n - v_n$ in~$\Drat$.
Note that the solution space of~$I_\bv$ is spanned by~$\bx^\bv$.  Set
\[  I = \Drat G \quad \text{and} \quad J = \bigcap_{\bv \in \bN^n_m \setminus \IE_0(G)} I_\bv.\]
Then the two left ideals $I$ and $J$ have no solution in common except the trivial one  because $\bv \in \bN^n_m \setminus \IE_0(G)$.
It follows from Lemma~\ref{LEM:exact} (iii) that
$$\sol_\bE(I \cap J) = \sol_\bE(I) \oplus \sol_\bE(J).$$ In particular, $\sol_\bE(I\cap J)$
has dimension $|\bN_m^n|$ over~$C_\bE$, because $\sol_\bE(I)$ and~$\sol_\bE(J)$ have dimensions $|\IE_0(G)|$ and~$|\bN_m^n|-|\IE_0(G)|$,
respectively. So~$\IE_0(I\cap J)=\bN^n_m$. Let~$M$ be a primitive Gr\"obner basis of~$I \cap J$. Then the origin is an ordinary point
of~$M$ by Lemma~\ref{LEM:wronskianrep}.

\medskip \noindent
\emph{Necessity.} Assume that $M$ is a left multiple of~$G$ and that the origin is an ordinary point of~$M$. Then we have that $\sol_{\bE}(G) \subset \sol_{\bE}(M)$.
We need to prove that $\sol_{\bE}(G)$ has a basis in $\bK[[\bx]]$.

Assume that $\{f_1, \ldots, f_\ell\} \subset \bK[[\bx]]$ is a basis of $\sol_\bE(M)$.
Since $\sol_{\bE}(G)$ is contained in~$\sol_{\bE}(M)$, every element of $\sol_{\bE}(G)$ is a linear combination of  $f_1, \ldots, f_\ell$ over $\bC_{\bE}$.
Assume that $f = z_1 f_1 + \ldots + z_\ell f_\ell$,
where $z_1, \ldots, z_\ell \in \bC_{\bE}$ are to be determined.
Let $G = \{ G_1, \ldots, G_k \}$.
The constraints
\[
 G_j(f) = 0, \quad j = 1, \ldots, k,
\]
on~$f$ are equivalent to the constraints
\[
 z_1 G_j(f_1) + \cdots + z_\ell G_j(f_\ell) = 0, \quad j = 1, \ldots, k,
\]
on constants $z_1$, \ldots, $z_\ell$.
By comparing the coefficients of $\bx^\bw (\bw \in \bN^n)$ in both sides of the above equations,
we derive a linear system $A \bz = {\bf 0}$, where the matrix~$A$ has infinitely many rows but $\ell$ columns,
and $\bz$ stands for the transpose of $(z_1, \ldots, z_\ell)$. Moreover, both $G_1, \ldots, G_\ell$ and $f_1, \ldots, f_\ell$ have coefficients in $\bK$. So $A$ is a matrix over $\bK$.
 We have that
\begin{equation} \label{EQ:basis}
f \in \sol_{\bE}(G) \quad \Longleftrightarrow  \quad A \bz = {\bf 0}.
\end{equation}
Set
$$\ker(A)=\left\{(c_1, \ldots, c_\ell)^t \mid A(c_1, \ldots, c_\ell)^t =(0, \ldots, 0)^t,  c_1, \ldots, c_\ell \in C_\bE \right\},$$
where~$(\cdot)^t$ stands for the transpose of a vector (matrix).
Since $f_1, \ldots, f_\ell$ are linearly independent over $C_\bE$, the linear independence of the elements in~$\sol_\bE(G)$ over~$C_\bE$
is equivalent to the linear independence of the corresponding vectors in $\ker(A)$ over~$C_\bE$.
Thus,  the dimension of $\ker(A)$ over $C_\bE$ is equal to $\dim_{C_\bE} \sol_\bE(G)$, which is~$d$ by Proposition~\ref{PROP:wn}.
It follows that the rank of~$A$ is equal to~$\ell-d$. Since all the coefficients of $A$ lie in $\bK$, there are $d$ vectors in the intersection of~$\ker(A)$ and $\bK^\ell$,
 which are linearly independent over $\bK$. These vectors are  also linearly independent over $C_\bE$. These vectors give rise to a basis of $\sol_{\bE}(G)$ in~$\bK[[\bx]]$.
The origin is an apparent singularity of~$G$ by Theorem~\ref{THM:chop}.
\end{proof}
Assume that the origin is a singularity of~$G$. By desingularizing the origin, we mean computing a left multiple~$M$ of~$G$ such that the origin is an ordinary point
of~$M$. The next corollary helps us to desingularize an apparent singularity.
\begin{cor} \label{COR:desing}
Let $G$ be a primitive Gr\"obner basis of finite rank. Assume that the origin is an apparent singularity of $G$.
Set~$m = \max_{\bu \in \IE_0(G)} |\bu|$. The the origin is an ordinary point of a primitive Gr\"obner basis of the left ideal
\[  \Drat G \cap \left( \bigcap_{(v_1, \ldots, v_n) \in \bN_m^n \setminus \IE_0(G)} \Drat \{ x_1 \pa_1 - v_1, \ldots, x_n \pa_n - v_n \} \right). \]
\end{cor}
\begin{proof}
It is direct from the proof on the sufficiency of the above theorem.
\end{proof}

\section{Desingularization and applications}

Let~$G$ be a primitive Gr\"obner basis with the origin being an apparent singularity.
We want to compute a left multiple $M$ of $G$ such that the origin is an ordinary point of $M$.
To this end,
we need to study $\IE_0(G)$ by Corollary~\ref{COR:desing}.

\subsection{Indicial ideals} \label{SUBSECT:indpol}

We extend the notion of indicial polynomials for linear ordinary differential operators to the D-finite case.

Let $\delta_i = x_i \pa_i$ be the Euler operator with respect to $x_i$, $i = 1, \ldots, n$.
The commutation rules in~$\Drat$ imply that~$\delta_i \delta_j = \delta_j \delta_i$ for all~$i,j \in \{1, \ldots, n\}.$
For~$\bu =(u_1, \ldots, u_n) \in \bN^n$, the symbol $\bdelta^\bu$ stands for the product~$\delta_1^{u_1} \cdots \delta_n^{u_n}$.

Recall that the $m$-th falling factorial~\cite[Section 3.1]{Kauers2011} of $x_i$ is
$$(x_i)^{\underline{m}} = x_i (x_i - 1) \cdots (x_i - m + 1),$$
where $m \in \bN$, $i = 1, \ldots, n$.
Moreover, $(x_i)^{\underline{0}} = 1$.
\begin{prop} \label{PROP:Euler}
The following assertions hold for Euler operators:
\begin{itemize}
 \item[(i)] For each $m \in \bN$ and $i \in \{1, \ldots, n \}$, $x_i^m \pa_i^m = (\delta_i)^{\underline{m}}$.
 \item[(ii)] For each $p \in \bK[\bx]$ and $\bx^{\bu} \in \TT(\bx)$, we have $p(\bdelta)(\bx^{\bu}) = p(\bu) \bx^{\bu}$.
\end{itemize}
\end{prop}
\begin{proof}
  See~\cite[Section 5.2]{Zhang2017}.
\end{proof}



Set $\bK[\by] = \bK[y_1, \ldots, y_n]$ to be the ring of usual commutative polynomials with indeterminates $y_1, \ldots, y_n$.

\begin{defn} \label{DEF:indpol}
Let a nonzero operator $P \in \Dpol$ be of order~$m$. Write
\begin{equation} \label{EQ:indpoleq}
\bx^{\bm} P = \sum_{\bv \in S} \bx^{\bv} \left( \sum_{| \bu | \leq m}  c_{\bu, \bv} \bdelta^{\bu} \right),
\end{equation}
where $\bm = (m, \ldots, m) \in \bN^n$ and~$S$ is a finite subset of~$\bN^n$.
Let $\bx^{\bv_0}$ be the minimal term among $\{x^{\bv} \mid \bv \in S \}$ with respect to~$\prec$ such that
$\sum_{| \bu | \leq m}  c_{\bu, \bv_0} \bdelta^{\bu}$ is nonzero.
We call
$$\sum_{| \bu | \leq m} c_{\bu, \bv_0} \by^{\bu} \in \bK[\by]$$
the \emph{indicial polynomial} of~$P$, and denote it by $\ind(P)$.
We further define $\ind(0) := 0$.
\end{defn}

By Proposition~\ref{PROP:Euler} (i), we may always write $\bx^{\bm} P$ in the form~\eqref{EQ:indpoleq}.
The above definition is compatible with the univariate case~\cite{Max2013, Saito1999},
and was already used in the multivariate setting~\cite[Definition 11]{Aroca2001}.

\begin{prop} \label{PROP:indpolroot}
Let $P$ be a nonzero element in $\Dpol$  and $f$
a formal power series solution of $P$ with initial exponent~$\bw$.
Then $\bw$ is a zero of~$\ind(P)$.
\end{prop}
\begin{proof}
Assume that $P$ is an operator of order~$m$ and write
$$\bx^{\bm} P = \sum_{\bv \in S} \bx^{\bv} \left( \sum_{| \bu | \leq m}  c_{\bu, \bv} \bdelta^{\bu} \right),$$
as given in Definition~\ref{DEF:indpol}.
By Proposition~\ref{PROP:Euler} (ii), we have
\[
\begin{array}{lll}
 \left( \bx^{\bm} P \right)(f) & = & \left[ \sum_{\bv \in S} \bx^{\bv} \left( \sum_{| \bu | \leq m}  c_{\bu, \bv} \bdelta^{\bu} \right) \right]
 \left( \bx^{\bw} + \text{ higher monomials in $\bx$} \right) \\
  & = & \bx^{\bv_0} \left( \sum_{| \bu | \leq m} c_{\bu, \bv_0} \bdelta^{\bu} \right) (\bx^{\bw}) + \text{ higher monomials in $\bx$} \\
  & = & \left( \sum_{| \bu | \leq m} c_{\bu, \bv_0} \bw^{\bu} \right) \bx^{\bv_0 + \bw} + \text{ higher monomials in $\bx$} \\
  & = & 0
\end{array}
\]
Thus, $\sum_{| \bu | \leq m} c_{\bu, \bv_0} \bw^{\bu} = 0,$
\ie, $\ind(P)(\bw) = 0$.
\end{proof}

\begin{ex} \label{EX:indpol}
Consider the Gr\"{o}bner basis $G = \{G_1, G_2 \}$ in $\bQ(x_1, x_2)[\pa_1, \pa_2]$, where
$G_1 = x_1 x_2 \pa_2 - x_1 x_2 \pa_1 + (x_2-x_1)$ and $ G_2 = x_1^2 \pa_1^2 - 2 x_1 \pa_1 + (2 + x_1^2).$
By computation, we find that $\ind(G_1) = y_2  - 1$, $\ind(G_2) = (y_1 -1) (y_1 - 2)$.
It is straightforward to verify that $G$ has two formal power series solutions
$$\{f_1 = x_1 x_2 \sin(x_1 + x_2), f_2 = x_1 x_2 \cos(x_1 + x_2) \},$$
with $\In(f_1) = x_1^2 x_2$ and $\In(f_2) = x_1 x_2$.
The corresponding initial exponents
$$\{ (2, 1), (1, 1) \}$$
are two zeros of $\ind(G_1)$ and $\ind(G_2)$.
\end{ex}


\begin{defn}
Let $I$ be a left ideal of $\Drat$.
We call
$$\left\{ \ind(P) \mid P \in I \cap \Dpol \right\}$$
the \emph{indicial ideal} of~$I$, and denote it by $\ind(I)$.
\end{defn}

\begin{thm} \label{THM:indideal}
Let $I$ be a left ideal of $\Drat$.
Then  $\ind(I)$ is an ideal in~$\bK[\by]$. Moreover, $\ind(I)$ is zero-dimensional if $I$ is D-finite.
\end{thm}
\begin{proof}
For any two $a, b \in \ind(G)$, there exist two operators $P, Q \in I$ such that $a = \ind(P)$ and $b = \ind(Q)$.
Let $u$ and $v$ be the respective orders of $P$ and~$Q$.
Set $\bu = (u, \ldots, u)$ and  $\bv = (v, \ldots, v)$.
Expressing~$\bx^\bu P$ and~$\bx^\bv Q$ as polynomials in~$\bx$ with coefficients in $\bK[\bdelta]$ placed on the right-hand side
of the powers of~$\bx$, we have
\[
\begin{array}{lll}
 \bx^{\bu} P & = & \bx^{\bs} \left( \sum_{| \bu | \leq u}  c_{\bu, \bs} \bdelta^{\bu} \right) + \text{ higher terms}, \\
 \bx^{\bv} Q & = & \bx^{\bt} \left( \sum_{| \bu | \leq v}  c_{\bu, \bt} \bdelta^{\bu} \right) + \text{ higher terms}.
\end{array}
\]
Thus, $a=\sum_{| \bu | \leq u}  c_{\bu, \bs} \by^{\bu}$ and~$b = \sum_{| \bu | \leq v}  c_{\bu, \bt} \by^{\bu}$.
Let $L = \bx^{\bt}(\bx^{\bu} P) + \bx^{\bs} (\bx^{\bv} Q)$, which belongs to $I$. Then
\[
 L = \bx^{\bs + \bt} \left( \sum_{| \bu | \leq u}  c_{\bu, \bs} \bdelta^{\bu}
 + \sum_{| \bu | \leq v}  c_{\bu, \bt} \bdelta^{\bu} \right) + \text{ higher terms}.
\]
Let $m$ be the order of $L$ and $\bm = (m, \ldots, m)$. Then
\[
 \bx^{\bm} L = \bx^{\bs + \bt + \bm} \left( \sum_{| \bu | \leq u}  c_{\bu, \bs} \bdelta^{\bu}
 + \sum_{| \bu | \leq v}  c_{\bu, \bt} \bdelta^{\bu} \right) + \text{ higher terms}.
\]
Thus, $a + b = \ind(L)$.

Assume that $r \in \bK[\by]$.
We want to prove that $r a \in \ind(G)$.
Since $r$ is a sum of monomials in $y_1, \ldots, y_n$, it suffices to prove that $r a \in \ind(G)$
for each term~$r$.
Assume that $r = \by^{\bw}$, where $\bw = (w_1, \ldots, w_n) \in \bN^n$.
Then
\[
 \bx^{\bu} P  =  \bx^{\bs} \left( \sum_{| \bu | \leq u}  c_{\bu, \bs} \bdelta^{\bu} \right) + \text{ higher terms},
\]
where $\bs = (s_1, \ldots, s_n) \in \bN^n$.
Let $H = \prod_{i = 1}^n (\delta_i - s_i)^{w_i} \bx^{\bu} P$, which belongs to~$I$,  and note that $(\delta_i-k) x_i^k = x_i^k \delta_i$ for
all $i \in \{1, \ldots, n\}$ and~$k \in \bN$. Then
\[
\begin{array}{lll}
 H & = &  \left( \prod_{i = 1}^n (\delta_i - s_i)^{w_i} \right) \bx^{\bs} \left( \sum_{| \bu | \leq u}  c_{\bu, \bs} \bdelta^{\bu} \right) + \text{ higher terms} \\
   & = &  \left( \prod_{i = 1}^n (\delta_i - s_i)^{w_i} x_i^{s_i} \right) \left( \sum_{| \bu | \leq u}  c_{\bu, \bs} \bdelta^{\bu} \right) + \text{ higher terms} \\
   & = &  \left( \prod_{i = 1}^n x_i^{s_i} \delta_i^{w_i} \right) \left( \sum_{| \bu | \leq u}  c_{\bu, \bs} \bdelta^{\bu} \right) + \text{ higher terms} \\
   & = & \bx^{\bs} \left( \bdelta^{\bw} \sum_{| \bu | \leq u}  c_{\bu, \bs} \bdelta^{\bu} \right) + \text{ higher terms}. \\
\end{array}
\]
Let $\tilde{m}$ be the order of $H$ and $\tilde{\bm} = (\tilde{m}, \ldots, \tilde{m})$. Then
\[
 \bx^{\tilde{\bm}} H = \bx^{\bs + \tilde{\bm}} \left( \bdelta^{\bw} \sum_{| \bu | \leq u}  c_{\bu, \bs} \bdelta^{\bu} \right) + \text{ higher terms}.
\]
Thus,  $r a = \ind(H)$. Consequently, $\ind(I)$ is an ideal in $\bK[\by]$.

Assume further that $I$ is D-finite. Then there exists a nonzero operator~$P$ of order~$m$ such that
$P \in I \cap \bK[\bx][\pa_1]$ (see, e.g., \cite[Proposition 2.10]{Christoph2009} for a proof).
By  Proposition~\ref{PROP:Euler} (i), we have
\[
\begin{array}{lll}
 x_1^m P & = & x_1^m \left(c_0 + c_1 \pa_1 + \cdots + c_m \pa_1^m \right) \\
       & = & c_0 x_1^m + c_1 x_1^{m - 1} \delta_1 + \cdots + c_m (\delta_1)^{\underline{m}} \\
       & = & \sum_{\bv \in T} \bx^{\bv} \left( \sum_{k \leq m}  c_{\bu, \bv} \delta_1^{k} \right)
\end{array}
\]
Thus, $\ind(P) \in \bK[y_1] \setminus \{ 0 \}$.
In the same vein, $\ind(I) \cap \bK[y_i]$ is nontrivial for all~$i$ with~$2 \le i \le n$.
By~\cite[Theorem 6, page 251]{Cox2015}, $\ind(I)$ is zero-dimensional.
\end{proof}

Given a Gr\"obner basis~$G$ of finite rank, the indicial ideal of~$\Drat G$ is simply denoted by $\ind(G)$.
By the last paragraph of the proof of Proposition~\ref{THM:indideal}, we can construct a sub-ideal~$J$ of $\ind(G)$ such that $J$ is zero-dimensional. However, the above proposition does not necessarily give access to a basis of $\ind(G)$.

\begin{defn} \label{DEF: indcandidate}
Let $G$ be a primitive Gr\"obner basis of finite rank.
Assume that $J$ is a zero-dimensional ideal contained in $\ind(G)$.
The set of nonnegative integer solutions of~$J$ is called
a set of \emph{initial exponent candidates} of~$G$.
\end{defn}

By Proposition~\ref{PROP:indpolroot}, the set of initial exponents of formal power series solutions of $G$
must be contained in a set of initial exponent candidates of~$G$.
Such a candidate set can be obtained by computing nonnegative integer solutions of some zero-dimensional system over~$\bK$.

\begin{ex} \label{EX:exponentcandidates1}
Consider the Gr\"{o}bner basis $G = \{G_1, G_2 \}$ from Example~\ref{EX:indpol}, where
$G_1 = x_1 x_2 \pa_2 - x_1 x_2 \pa_1 + (-x_1 + x_2)$ and $G_2 = x_1^2 \pa_1^2 - 2 x_1 \pa_1 + (2 + x_1^2).$
By computation, we find that $\ind(G_1) = y_2  - 1$, $\ind(G_2) = (y_1 -1) (y_1 - 2)$.
By the above definition, the set
$\{ (2, 1), (1, 1) \}$
is a set of initial exponent candidates of~$G$. Actually, $(2, 1)$ and $(1, 1)$ are initial exponents
of the following formal power series solutions
$x_1 x_2 \sin(x_1 + x_2) \text{ and } x_1 x_2 \cos(x_1 + x_2),$
respectively.
\end{ex}

The following example shows that initial candidates of~$G$ do not necessarily
give rise to formal power series solutions of~$G$.

\begin{ex} \label{EX:exponentcandidates2}
Consider the Gr\"{o}bner basis in $\bQ(x_1, x_2)[\pa_1, \pa_2]$:
\[
\begin{array}{lll}
 G & = & \{G_1, G_2 \} \\
   & = & \{x_1 x_2 \pa_2 + (-x_1^2 + 2 x_1 x_2) \pa_1 - 2 x_2, (x_1^3 - x_1^2 x_2) \pa_1^2 + 2 x_1 x_2 \pa_1 - 2 x_2 \}
\end{array}
\]
By computation, we find that $\ind(G_1) = y_2 - y_1$ and $\ind(G_2) = (y_1 - 1) y_1$.
Thus, a set of initial exponent candidates of $G$ is
\[
 S = \{(0, 0), (1, 1) \}
\]
Actually, $\sol_{\bE}(G)$ is spanned by $\{\frac{x_1}{x_1 - x_2}, x_1 x_2 \}$.
In this case, $(1, 1)$ is the initial exponent of $x_1 x_2$.
However, $(0, 0)$ is not the initial exponent of any formal power series solution of~$G$.
\end{ex}
\subsection{Desingularization}

We present two algorithms: one is for removing apparent singularities, and the other is for detecting whether a singularity is apparent.
\begin{algo} \label{ALGO:desingularization1}
Given a primitive Gr\"obner basis $G$ of finite rank such that the origin is an apparent singularity of $G$,
compute a left multiple $M$ of $G$ such that the origin
is an ordinary point of~$M$.
\begin{itemize}
\item[(1)] Set $d:=\rank(G)$.
\item [(2)] Compute a set of initial exponent candidates $S$ by Theorem~\ref{THM:indideal}.
 \item [(3)] For each $B \subset S$ with $|B|=d$,
 \begin{itemize}
 \item[(3.1)] set $m:= \max_{\bu \in B} |\bu|$;
 \item[(3.2)]compute a primitive Gr\"obner basis $M_B$ of the left ideal
 \[  \Drat G \cap \left( \bigcap_{(v_1, \ldots, v_n) \in \bN_m^n \setminus B} \Drat \{ x_1 \pa_1 - v_1, \ldots, x_n \pa_n - v_n \} \right); \]
 \item[(3.3)] if the origin is an ordinary point, then return $M_B$.
 \end{itemize}
\end{itemize}
\end{algo}
The above algorithm evidently terminates.  We have that~$\IE_0(G) \subset S$ and $|\IE_0(G)|=d$, since the origin is an apparent singularity of~$G$.
It follows from Corollary~\ref{COR:desing} that there exists a subset $B$ of~$S$ with~$|B|=d$ such that the origin is an ordinary point of~$M_B$.
So the above algorithm is correct.

\begin{ex}
Consider the Gr\"{o}bner basis in Example~\ref{EX:appsin1}:
\[
 G = \{x_2 \pa_2 + \pa_1 - x_2 - 1, \pa_1^2 - \pa_1 \},
\]
where $\sol_{\bE}(G)$ is spanned by $\{\exp(x_1 + x_2), x_2 \exp(x_2)\}$ with
initial exponents
$$B = \{ (0, 0), (0, 1) \}.$$
In this case, the origin is an apparent singularity of~$G$.

Note that $\bN^2_1 = \{ (i, j) \in \bN^2 \mid i + j \leq 1 \} = \{ (0, 0), (1, 0), (0, 1) \}$.
Then
$$\bN^2_1 \setminus B = \{ (1, 0) \}.$$
Let $M$ be another Gr\"{o}bner basis with
\[
 \Drat M = \Drat G \cap \Drat \{x_1 \pa_1 - 1, \pa_2 \}.
\]
We find that $\HC(M) = \{ 1 - x_1 - x_1 x_2\}$.
It follows from Definition~\ref{DEF:op} that~$M$ is a left multiple of $G$ for which
the origin is an ordinary point.
\end{ex}

\begin{ex}
Consider the Gr\"{o}bner basis from Example~\ref{EX:appsin2}:
\[
G = \{x_2^2 \pa_2 - x_1^2 \pa_1 + x_1 - x_2, \pa_1^2 \},
\]
where $\sol_{\bE}(G)$ is spanned by $\{x_1 + x_2, x_1 x_2\}$.
In this case, the origin is an apparent singularity of~$G$.

A candidate set $B$ of indicial exponents is
$B = \{(1, 0), (1, 1) \} .$
Then
$$\bN^2_2 \setminus B = \{ (0, 0), (0, 1), (2, 0), (0, 2) \}.$$
Let $M$ be another Gr\"{o}bner basis with
\[
 \Drat M = \Drat G \cap \left(\bigcap_{(s, t) \in \bN^2_2 \setminus B} \Drat \{ x_1 \pa_1 - s, x_2 \pa_2 - t \} \right)
\]
We find that
\[
 M = \{\pa_1^3,  \pa_1^2 \pa_2, \pa_1 \pa_2^2, \pa_2^3 \}.
\]
The origin is an ordinary point of $M$.
\end{ex}

The next algorithm is a direct application of Algorithm~\ref{ALGO:desingularization1}.
\begin{algo} \label{ALGO:desingularization2}
Given a primitive Gr\"obner basis $G$ of finite rank such that the origin is a singularity of $G$,
determine whether the origin is an apparent one, and return a left multiple $M$ of $G$ such that the origin
is an ordinary point of~$M$ when the origin is an apparent singularity.
\begin{itemize}
\item[(1)] Set $d:=\rank(G)$.
\item [(2)] Compute a set of initial exponent candidates $S$ by Theorem~\ref{THM:indideal}. If $|S| < d$, then return \lq\lq the origin is not an apparent singularity\rq\rq.
 \item [(3)] For each $B \subset S$ with $|B|=d$,
 \begin{itemize}
 \item[(3.1)] set $m:= \max_{\bu \in B} |\bu|$;
 \item[(3.2)]compute a primitive Gr\"obner basis $M_B$ of the left ideal
 \[  \Drat G \cap \left( \bigcap_{(v_1, \ldots, v_n) \in \bN_m^n \setminus B} \Drat \{ x_1 \pa_1 - v_1, \ldots, x_n \pa_n - v_n \} \right); \]
 \item[(3.3)] if the origin is an ordinary point, then return $M_B$.
 \end{itemize}
 \item[(4)] Return \lq\lq the origin is not an apparent singularity\rq\rq.
\end{itemize}
\end{algo}
The above algorithm clearly terminates. If the candidate set~$S$ has less than $d$ many elements, then the solution space of~$G$ in $\bE$ cannot be spanned
by formal power series. Thus, the origin is not an apparent singularity. The rest of the above algorithm is correct by Algorithm~\ref{ALGO:desingularization1}.

%
%

\begin{ex} \label{EX:detectappsin1}
Consider the Gr\"{o}bner basis from Example~\ref{EX:exponentcandidates2}:
\[
\begin{array}{lll}
 G & = & \{G_1, G_2 \} \\
   & = & \{x_1 x_2 \pa_2 + (-x_1^2 + 2 x_1 x_2) \pa_1 - 2 x_2, (x_1^3 - x_1^2 x_2) \pa_1^2 + 2 x_1 x_2 \pa_1 - 2 x_2 \}.
\end{array}
\]
Then $\rank(G) = 2$ and the origin is a singularity of~$G$.
By computation, we find that $\ind(G_1) = y_2 - y_1$ and $\ind(G_2) = (y_1 - 1) y_1$.
Thus, a set of initial exponent candidates of $G$ is
\[
 S = \{(0, 0), (1, 1) \}
\]
Let $B = S$. Then $\bN_2^2 \setminus B = \{ (1, 0), (0, 1), (2, 0), (0, 2) \}.$
Let $M$ be another Gr\"{o}bner basis with
\[
 \Drat M = \Drat G \cap \left(\bigcap_{(s, t) \in U_2 \setminus B} \Drat \{ x_1 \pa_1 - s, x_2 \pa_2 - t \} \right).
\]
We find that
\[
 \HC(M) = \{x_1^4 - 3 x_1^3 x_2 + 3 x_1^2 x_2^2 - x_1 x_2^3, - x_1^3 + 3 x_1^2 x_2 - 3 x_1 x_2^2 + x_2^3 \}.
\]
Thus, the origin is a singularity of~$M$.
By Theorem~\ref{THM:rmappsin}, we conclude that the origin is not an apparent singularity of~$G$.
Actually, $\sol_{\bE}(G)$ is spanned by
$$\left\{\frac{x_1}{x_1 - x_2}, x_1 x_2 \right\}.$$
\end{ex}

\begin{ex} \label{EX:detectappsin2}
Consider the Gr\"{o}bner basis in $\bQ(x_1, x_2)[\pa_1, \pa_2]$:
\begin{alignat*}1
 G & =  \{G_1, G_2, G_3 \} \\
   & =  \{(x_1 - x_2) \pa_1^2 - x_1 x_2 \pa_2 + x_1 x_2 \pa_1 + (x_1 - x_2), \\
   & \hphantom{{}=\{}   (x_1 - x_2) \pa_1 \pa_2 + (-1 - x_1 x_2) \pa_2 + (1 + x_1 x_2) \pa_1 + (x_1 - x_2), \\
   & \hphantom{{}=\{}   (x_1 - x_2) \pa_2^2 - x_1 x_2 \pa_2 + x_1 x_2 \pa_1 + (x_1 - x_2) \}.
\end{alignat*}
Then $\rank(G) = 3$ and the origin is a singularity of~$G$.
From the three indcial polynomials $\ind(G_1) = (y_1 - 1) y_1$, $\ind(G_2) = y_2 (y_1 - 1)$
and $\ind(G_3) = (y_2 - 1) y_2$, we find that a set of initial exponent candidates of $G$ is
\[
 S = \{(0, 0), (1, 0), (1, 1) \}.
\]
Let $B = S$. Then $\bN^2_2 \setminus B = \{ (0, 1), (2, 0), (0, 2) \}.$
Let $M$ be another Gr\"{o}bner basis with
\[
 \Drat M = \Drat G \cap \left(\bigcap_{(s, t) \in \bN^2_2 \setminus B} \Drat \{ x_1 \pa_1 - s, x_2 \pa_2 - t \} \right).
\]
We find that the origin is an ordinary point of~$M$.
By Theorem~\ref{THM:rmappsin}, we conclude that the origin is an apparent singularity of~$G$.
Actually, $\sol_{\bE}(G)$ is spanned by
$$\{\sin(x_1 + x_2), \cos(x_1 + x_2), x_1 x_2 \}.$$
\end{ex}

\subsection{A heuristic method for desingularization}

For a nonzero operator $L\in \bK[x_1][\partial_1]$ with apparent singularities,
the randomized algorithm in~\cite{Chen2016} computes a desingularized operator for $L$
by taking the least common left multiple of the operator~$L$ with a random operator of appropriate order with constant coefficients.
This algorithm has been proved to obtain a correct desingularized operator for $L$ with probability one, and is more efficient than deterministic algorithms.
We now extend this randomized technique to the case of several variables. To this end, we need two lemmas about determinants.

\begin{lemma}\label{LEM:indet}
Let~$A=(a_{ij})$ be a full rank $(k+d) \times d$ matrix over~$\bK$
and~$(y_{ij})$ be a $(k+d) \times k$ matrix whose entries are distinct indeterminates.
Then the determinant
\[
 \Delta=\left|
  \begin{array}{cccccc}
    a_{1,1} & \cdots  & a_{1, d}  & y_{11}  &\cdots & y_{1k} \\
    a_{2,1} & \cdots  & a_{2, d}  & y_{21} &\cdots & y_{2k} \\
    \vdots & \ddots & \vdots  & \vdots & \ddots & \vdots \\
 a_{k+d,1} & \cdots  & a_{k+d,d}  & y_{k+d,1} &\cdots & y_{k+d,k}\\
  \end{array}
\right|
\]
is equal to a nonzero polynomial of the form
\begin{equation} \label{EQ:poly}
\sum_{(i_1, \ldots, i_k) \in S}  \alpha_{i_1, \ldots, i_k} y_{i_1,1} \cdots y_{i_k,k},
\end{equation}
 where $S$ is a nonempty subset of~$\bN^k_{k+d}$, and $\alpha_{i_1, \ldots, i_k} \in \bK$ is nonzero for every
 $(i_1, \ldots, i_k) \in S$.
\end{lemma}
\begin{proof}
Since~$A$ is of full rank, there exists a $d \times d$ nonzero minor in~$A$. Without loss of generality,
we may assume that the minor consists of the first $d$ rows and the first $d$ columns.
Setting $y_{ij}=0$ for all~$i,j$ with~$1 \le i \le d$ and~$1 \le j \le k$,
we transform the determinant $\Delta$ to
\[
 \left|
  \begin{array}{cccccc}
    a_{1,1} & \cdots  & a_{1, d}  &  0  &\cdots & 0  \\
    \vdots  & \ddots  & \vdots    &  \vdots & \ddots & \vdots \\
    a_{d,1} & \cdots  & a_{d, d}  &  0 &\cdots &  0 \\
    a_{d+1,1} & \cdots  & a_{d+1, d}  &  y_{d+1,1} &\cdots &  y_{d+1,k} \\
    \vdots & \ddots & \vdots  & \vdots & \ddots & \vdots\\
 a_{d+k,1} & \cdots  & a_{d+k,d}  & y_{d+k,1} &\cdots & y_{d+k,k}\\
  \end{array}
\right|,
\]
which is nonzero. So $\Delta$ is nonzero. Collecting the like terms of the determinant. we see that $\Delta$ is of the form~\eqref{EQ:poly}.
\end{proof}
\begin{lemma}\label{LEM:mon}
Let~$A=(a_{ij})$ be a full rank $(d+k) \times d$ matrix over~$\bK$, $Z_1, \ldots, Z_k$ be mutually disjoint sets of indeterminates.
Let~$b_{1,j}, \ldots, b_{d+k,j}$ be distinct terms in the indeterminates belonging to~$Z_j$, $j=1, \ldots, k$.
Then the determinant
\[
D = \left|
  \begin{array}{cccccc}
    a_{1,1} & \cdots  & a_{1, d}  & b_{1,1}  &\cdots & b_{1,k} \\
    a_{2,1} & \cdots  & a_{2, d}  & b_{2,1} &\cdots & b_{2,k} \\
    \vdots & \ddots & \vdots  & \vdots & \ddots & \vdots \\
 a_{d+k,1} & \cdots  & a_{d+k,d}  & b_{d+k,1} &\cdots & b_{d+k,k}\\
  \end{array}
\right|
\]
is a nonzero polynomial in~$\bK[Z_1 \cup \cdots \cup Z_n]$.
\end{lemma}
\begin{proof}
By~\eqref{EQ:poly},
\begin{equation} \label{EQ:expr}
D  = \sum_{(i_1, \ldots, i_n) \in W} \alpha_{i_1, \ldots, i_n} b_{i_1,1} \cdots b_{i_n,n}.
\end{equation}
For two distinct elements $(i_1, \ldots, i_n), (j_1, \ldots, j_n) \in W$,  the two terms $b_{i_1,1} \cdots b_{i_n,n}$
and $b_{j_1,1} \cdots b_{j_n,n}$ are also distinct by the definition of $b_{ij}$'s. Hence, there are no like terms to be collected in the right-hand side of~\eqref{EQ:expr}, which implies that $D$ is nonzero.
\end{proof}

\begin{thm} \label{THM:random}
Let~$G$ be a primitive Gr\"obner basis of rank $d$.
Assume that the origin is an apparent singularity of~$G$, and~$f_1, \ldots, f_d$ be $\bK$-linearly independent formal power series solutions
of~$G$ with distinct initial exponents~$\bu_1$, \ldots, $\bu_d$, respectively.
Set $m=\max_{1 \le i \le d}|\bu_i|$ and $\bN_m^n \setminus \IE_0(G) = \{ \bu_{d+1}, \ldots, \bu_\ell\}$.
For each~$j \in \{1, \ldots, \ell-d\}$, let $f_{d+j}$ be the formal power series expansion of
\[ \exp\left(z_{1,j}x_1 + \cdots + z_{n,j}x_n\right) \]
at the origin, where $z_{1j}, \ldots, z_{nj}$ are distinct constant indeterminates. Let $A$ be the $\ell \times \ell$ matrix
whose element at the $i$th row and $j$th column is the formal power series $\bpa^{\bu_i} f_j$ evaluated at the origin, $i,j \in \{1, \ldots, \ell\}$.
Then
\begin{itemize}
\item[(i)] $\det(A)$ is a nonzero polynomial in~$\bK[z_{1,1}, \ldots, z_{n,1}, \ldots, z_{1,\ell-d}, \ldots, z_{n,\ell-d}].$
\item[(ii)] For~$1 \le i \le n$ and~$1 \le j \le \ell-d$, let~$c_{i,j}$ be elements of~$\bK$.
If $\det(A)$ does not vanish at $(c_{1,1}, \ldots, c_{n,1}, \ldots, c_{1 , \ell-d}, \ldots, c_{n,\ell-d})$,
then the origin is an ordinary point of the primitive Gr\"obner basis $M$ of
\[ \Drat G \cap \left(\bigcap_{j=1}^{\ell-d} \Drat \left\{ \pa_1 - c_{1,j}, \ldots, \pa_n - c_{n,j} \right\}\right). \]
\end{itemize}
\end{thm}
\begin{proof}
We need two ring homomorphisms in the proof.

Let $R = \bK[z_{1,1}, \ldots, z_{n,1}, \ldots, z_{1,\ell-d}, \ldots, z_{n,\ell-d}].$  We define $\phi$ to be
the homomorphism from~$R[[\bx]]$ to~$R$ that takes the constant term of a formal power series in~$\bx$, which extends
the homomorphism from~$\bK[[\bx]]$ to~$\bK$ defined in Section~\ref{SUBSECT:fps}. For every nonzero formal power series
$f \in R$ with initial exponents $\bv$, we have  
\begin{equation} \label{EQ:initial}
\text{$\forall \, \bw \in \bN^n$ with $\bw \prec \bv,$} \,\, \phi\left(\bpa^\bw(f)\right) = 0
\quad \text{and} \quad \phi\left(\bpa^\bv(f)\right) \neq 0.
\end{equation}
by~\eqref{EQ:zf}. Let~$\psi: R \longrightarrow \bK$ be the substitution
that maps~$z_{ij}$ to~$c_{ij}$ for every $i \in \{1, \ldots, n\}$ and~$j \in \{1, \ldots, \ell-d\}$. Then~$\psi$ is a ring homomorphism. We extend~$\psi$ to a homomorphism from~$R[[\bx]]$ to~$\bK[[\bx]]$ by the rule $\psi(x_i)=x_i$. $i=1,\ldots, n$. The extended homomorphism is also denoted by~$\psi$.

(i) Without loss of generality, we order the initial exponents $\bu_1, \ldots, \bu_d$ increasingly with respect to~$\prec$.
Then the submatrix consisting of the first $d$ rows and first $d$ columns in~$A$ is in a lower triangular form whose elements
in the diagonal are all nonzero by~\eqref{EQ:initial}.
Thus, the first $d$ columns of $A$ are linearly independent over~$\bK$. Let~$\bz_j=(z_{1,j}, \ldots, z_{n,j})$, $j=1, \ldots, \ell-d$.
Then the $(d+j)$th column of~$A$ consists of $\bz_j^{\bu_1}, \ldots, \bz_j^{\bu_\ell}$, which are distinct terms in~$\bz_j$. Thus,~$\det(A)$ is nonzero by Lemma~\ref{LEM:mon}.

(ii) Let~$g_j = \psi(f_j)$ for~$j=1, \ldots, \ell$. Then the element at the $i$th row and $j$th column of $\psi(A)$ is the image of~$\phi\left(\bpa^\bu_i(g_j)\right)$ for all $i,j \in \{1, \ldots, \ell\}$£¬
because $\psi \circ \bpa^{\bu_i} = \bpa^{\bu_i} \circ \psi$. Moreover,~$\det(\psi(A))$ is nonzero because $\det(A)$ does not vanish at $(c_{1,1}, \ldots, c_{n,1}, \ldots, c_{1, \ell-d}, \ldots, c_{n, \ell-d})$.
Let~$B$ be the $\ell \times \ell$ matrix whose element at the $i$th row and $j$th column is equal to $\bpa^{\bu_i} g_j$ for all $i,j \in \{1, \ldots, \ell\}$.
Then $\phi(B) = \psi(A)$.  Thus, $\det(\phi(B))$ is nonzero, and so is $\det(B)$.
It follows that $g_1, \ldots, g_\ell$ are linearly independent over $\bK$ by Theorem~1 in~\cite[Chapter II]{Kolchin1973}
(see also Section~\ref{SUBSECT:sw}).

Set $I = \Drat G$ and $I_j =  \Drat \left\{ \pa_1 - c_{1,j}, \ldots, \pa_n - c_{n,j} \right\}$ for all~$j$ in $\{1, \ldots, \ell-d\}$. Then $g_1, \ldots, g_d$ form a basis of $\sol_\bE(I)$, because $g_i=f_i$, $i=1, \ldots, d$,
and  $g_{d+j}$ spans $\sol_\bE(I_j)$, because $g_{d+j}$ is the formal power series expansion
of
$\exp\left(c_{1,j}x_1 + \cdots + c_{n,j}x_n\right)$ at the origin
for all~$j \in \{1, \ldots, \ell-d\}$.
It follows from Lemma~\ref{LEM:exact} (iii) that $g_1, \ldots, g_\ell$ form a basis of $\sol_\bE(M)$.

To prove that the origin is an ordinary point of $M$, it suffices to find a basis of $\sol_\bE(M)$ in $\bK[[\bx]]$ whose initial exponents are exactly
the elements of $\bN^n_m$ by Lemma~\ref{LEM:wronskianrep}. Since $g_1, \ldots, g_\ell$ are linearly independent over $\bK$,
there exists an $\ell \times \ell$ matrix $C$ over $\bK$  such that
\[   (h_1, \ldots, h_\ell) = (g_1, \ldots, g_\ell) C, \]
in which $h_1, \ldots, h_\ell$ have distinct initial exponents.

Set~$H=BC$. Then $H$ is the $\ell \times \ell$ matrix
whose element at the $i$th row and $j$th column is equal to~$\bpa^{\bu_i} h_j$. Moreover,~$\phi(H)$ is of full rank since
 $\phi(H)$ is equal to  $\phi(B) C$. Suppose that there exists $j \in \{1, \ldots, \ell\}$ such that its initial exponent does not belong
 to~$\bN^n_m$. Then it is higher than any element in~$\bN^n_m$, because $\prec$ is graded. It follows from~\eqref{EQ:initial}
 that the $j$th column of~$\phi(H)$ is a zero vector by~\eqref{EQ:initial}, a contradiction. Therefore, the initial exponents
 of $h_1$, \ldots, $h_\ell$ are exactly the elements of~$\bN_m^n$.
\end{proof}


\begin{algo} \label{ALGO:rand}
Given a primitive Gr\"obner basis $G$ of finite rank such that the origin is an apparent singularity of $G$,
compute a left multiple $M$ of $G$ such that the origin is an ordinary point of~$M$ or return ``fail''.

\begin{itemize}
\item[(1)] Set $d:=\rank(G)$.
\item [(2)] Compute a set of initial exponent candidates $S$ by Theorem~\ref{THM:indideal}.
 \item [(3)] For each $B \subset S$ with $|B|=d$,
 \begin{itemize}
 \item[(3.1)] set $m:= \max_{\bu \in B} |\bu|$ and $\ell:=|\bN^n_m|$;
 \item[(3.2)] choose a point $\bc = (c_{1,1}, \ldots, c_{n,1}, \ldots, c_{1,\ell-d}, \ldots, c_{n,\ell-d}) \in \bK^{n(\ell-d)}$;
 \item[(3.3)]compute the primitive Gr\"obner basis $M_B$ of the left ideal
 \[  \Drat G \cap \left( \bigcap_{j=1}^{\ell-d} \Drat \{ \pa_1 - c_{1,j}, \ldots, \pa_n - c_{n,j} \} \right); \]
 \item[(3.4)] if the origin is an ordinary point, then return $M_B$.
 \end{itemize}
 \item[(4)] return ``fail''.
\end{itemize}
\end{algo}
The above algorithm clearly terminates.  If $B=\IE_0(G)$ and $\det(A)$ given in Theorem~\ref{THM:random} does not vanish at $\bc$,
then the origin is an ordinary point of~$M_B$ by Theorem~\ref{THM:random}. So it does not return ``fail'' unless $\bc$ lies
on the variety defined by $\det(A)=0$. In this sense, we say that the above algorithm succeeds with probability one.
An advantage of the above algorithm is that it is more efficient to compute a Gr\"obner basis of the intersection of several
left ideals, most of which are generated by first-order operators with constant coefficients.
Another advantage is that this algorithm is likely to remove all apparent singularity, not just the origin,
because almost all choices of $c_{i,j}$ will also work for apparent singularities at almost any other point.
On the other hand,
it is not convenient to apply Theorem~\ref{THM:random} to determine whether the origin is an apparent singularity,
because the above algorithm will always return ``fail'' if the origin is a singularity but not an apparent one.


\begin{ex} \label{EX:randomappsin}
Consider the Gr\"{o}bner basis in Example~\ref{EX:appsin1}:
\[
\begin{array}{lll}
 G & = & \{G_1, G_2 \} \\
   & = & \{x_2 \pa_2 + \pa_1 - x_2 - 1, \pa_1^2 - \pa_1\}.
\end{array}
\]
In this case, $n= 2$ and $d := \rank(G) = 2$ and the origin is an apparent singularity of $G$. Set
\[
\begin{array}{lll}
 P & = & (\frac{1}{x_2} \pa_2 - \frac{1}{x_2^2} \pa_1 - \frac{1}{x_2})G_1 + \frac{1}{x_2^2}G_2  \\
   & = & \pa_2^2 - 2 \pa_2 + 1.
\end{array}
\]
We find that $\ind(P) = y_2 (y_2 - 1)$, $\ind(G_1) = y_1$
and $\ind(G_2) = y_1 (y_1 - 1)$.
Thus, a set of initial exponent candidates of $G$ is $S = \{(0, 0), (0, 1) \}.$
Set $B = S$ and $\ell = |\bN_1^2| = 3$.
Choose $\bc = (19, 23) \in \bK^{n(\ell -d)} = \bK^2$.
Let $M_B$ be the primitive Gr\"{o}bner basis of the left ideal
\[
 \Drat G \cap \Drat \{\pa_1 - 19, \pa_2 - 23 \}.
\]
We find that
$\HC(M_B) = \{9 + 11 x_2\}.$
It follows from Definition~\ref{DEF:op} that~$M_B$ is a left multiple of $G$ for which
the origin is an ordinary point.
\end{ex}

In the above example, if we take the roots of $\{ \ind(P), \ind(G_2) \}$
as a set of initial exponent candidates of $G$, then it strictly contains the set of initial exponents of $G$.

\subsection{Formal power series solutions at apparent singularities} \label{SUBSECT:fpsappsin}

Let $G$ be a primitive Gr\"{o}bner basis. Assume that the origin is an
apparent singularity of $G$. Then $\sol_\bE(G)$ has a basis of formal power series. In this subsection, we present an
algorithm for computing such a basis truncated at a given total degree.

Consider a fixed series $f = \sum_{\bu \in \bN^n} \frac{c_{\bu}}{\bu !} \bx^{\bu} \in \bK[[\bx]]$.
For each $m \in \bN$, set
\[
(f)_m =  \sum_{|\bu| \leq m} \frac{c_{\bu}}{\bu !} \bx^{\bu} \in \bK[\bx].
\]
We call $(f)_m$ the {\em $m$-th truncated power series}~\cite[page 35]{Kauers2011} of~$f$.
As a matter of notation, we write $T_m \subset \bK[\bx]$ for the ideal generated by~$\{\bx^\bu \mid |\bu| = m \}$.


Our idea is based on the proof of the necessity in Theorem~\ref{THM:rmappsin}. Assume that $\rank(G)=d$.
We desingularize the origin by a left multiple $M$ of~$G$, and then compute a
power series basis~$f_1, \ldots, f_\ell$ of~$\sol_\bE(M)$. Every power series solution of~$G$
is a linear combination of~$f_1, \ldots, f_\ell$ over~$\bK$ because $\sol_\bE(G) \subset \sol_\bE(M)$.
More precisely, every solution of the system $A \bz = {\bf 0}$ in~\eqref{EQ:basis} gives rise to a power series solution of~$G$ and vise versa.
Assume that the rows of $A$ are labelled by the elements of~$\bN^n$ increasingly with respect to~$\prec$. And let $A_m$ be the matrix
consisting of all rows labelled by the elements whose total degree are not higher than $m$.
Then there exists $m \in \bN$ such that $A_m$ has a right kernel of dimension~$d$,
because the solution space of $A \bz = \bzero$ in~\eqref{EQ:basis} is of dimension $d$ over~$\bK$.
So we just need to find the matrix $A_m$ incrementally using
$(f_1)_{m + r}$, \ldots, $(f_\ell)_{m + r}$
until $\rank(A_m) = \ell-d$, and then to compute a basis of the right kernel of~$A_m$,
where~$r$ is the maximal order of elements in $G$.
 This idea is encoded in the following algorithm.

\begin{algo} \label{ALGO:fpsappsin}
Given $m \in \bN$, and a primitive Gr\"obner basis $G$ of rank~$d$
such that the origin is an apparent singularity, compute polynomials $p_1, \ldots, p_d \in \bK[\bx]$
such that there exist $\bK$-linearly independent power series solutions $g_1, \ldots, g_d$ of~$G$ with the property
\[  p_1 = (g_1)_m, \ldots, p_d = (g_d)_m. \]
\begin{itemize}
\item[(1)] [Desingularize]
By Algorithm~\ref{ALGO:desingularization1}, compute a left multiple $M$ of $G$ such that the origin is an ordinary point of $M$.
And set~$\ell := \rank(M)$.
\item[(2)] [Initialize] Let $r$ be the maximal order of elements in $G$.
Set~$s = m$ and $z_1, \ldots, z_\ell$ to be constant indeterminates,
\item[(3)] [Construct a matrix of the maximal rank incrementally] Repeat
\begin{itemize}
\item[(3.1)] Set $s=s+1$.
\item[(3.2)] By formula~\eqref{EQ:ansatz}, compute a truncated power series in the form
\begin{equation} \label{EQ:par}
  \left( \sum_{\bu \in \bN^n} \frac{c_\bu}{\bu !} \bx^\bu \right)_{s + r},
\end{equation}
in which $c_\bu$ is an arbitrary constant for every~$\bu \in \PE(M)$.
\item[(3.3)] For each $\bu \in \PE(M)$, specialize $c_\bu =1$ and $c_{\bu^\prime}=0$ with~$\bu^\prime \neq \bu$ in \eqref{EQ:par}
to obtain polynomials
$h_1$, \ldots, $h_\ell$.
\item[(3.4)] Construct a matrix~$A_s$ with~$\ell$ columns over~$\bK$ such that its right kernel is equal to the solution space of the linear system
\[   G_t(z_1 h_1 + \cdots + z_\ell h_\ell) \equiv  0 \mod T_{s + 1}, \ \ \,\, t=1, \ldots, k. \]
\end{itemize}
Until~$\rank(A_s)=\ell-d$.
\item[(4)] [Compute truncated solutions]
Find a basis $\left(c_{1,1}, \ldots, c_{\ell, 1}\right)^t$, \ldots, $\left(
c_{1,d},  \ldots,  c_{\ell, d} \right)^t$ for the right kernel of~$A_s$ and set
\[    p_j = \left(\sum_{i=1}^\ell c_{i,j} h_j \right)_m, \quad j=1, \ldots, d. \]
\item[(5)] Return $p_1, \ldots, p_d$.
\end{itemize}
\end{algo}

\begin{ex} \label{EX:fpsappsin}
Consider the Gr\"{o}bner basis in Example~\ref{EX:appsin1}:
\[
\begin{array}{lll}
 G & = & \{G_1, G_2 \} \\
   & = & \{x_2 \pa_2 + \pa_1 - x_2 - 1, \pa_1^2 - \pa_1\}.
\end{array}
\]
In this case, $\rank(G) = 2$ and the origin is an apparent singularity of $G$.
We compute 2nd order truncated power series of a basis of $\sol_{\bE}(G)$.

(1) Let $M$ be the primitive Gr\"{o}bner basis of the left ideal
\[
 \Drat G \cap \Drat \{x_1 \pa_1 - 1, \pa_2 \}.
\]
We find that the origin is an ordinary point of $M$ and $\rank(M) = 3$.

(2) Let r = 2 be the maximal order of elements in G.
Set $s = 2$ and $z_1, z_2, z_3$ to be constant indeterminates.

(3) By formula~\eqref{EQ:ansatz},
we obtain the following $4$-th truncated power series of a basis of $\sol_{\bE}(M)$:
\[
\begin{array}{lll}
 h_1 & = & \left(\exp(x_1 + x_2) - x_1 - x_2 \exp(x_2)  \right)_4, \\
 h_2 & = & x_1, \\
 h_3 & = & \left( x_2 \exp(x_2) \right)_4.
\end{array}
\]
We find a matrix $A_2 \in \bK^{12 \times 3}$ of rank $1$ such that its right kernel is equal to the solution space of the linear system
\[   G_t(z_1 h_1 + z_2 h_2 + z_3 h_3) \equiv  0 \mod T_{3}, \ \ \,\, t = 1, 2. \]

(4) We find that the right kernel of $A_2$ has a basis $(1, 1, 0)^t, (0, 0, 1)^t$, and set
\[
\begin{array}{lll}
 p_1 & = & (h_1 + h_2)_2 = \left(\exp(x_1 + x_2) - x_2 \exp(x_2) \right)_2, \\
 p_2 & = & (h_3)_2 =  \left( x_2 \exp(x_2) \right)_2.
\end{array}
\]

(5) Return $p_1, p_2$.

Actually, $\sol_{\bE}(G)$ has a basis $\{\exp(x_1 + x_2), x_2 \exp(x_2) \}$.
So, $p_1, p_2$ are indeed 2nd order truncated power series of a basis of $\sol_{\bE}(G)$.
%
\end{ex}

\bibliographystyle{abbrv}
\def\cprime{$'$}

\end{document}